\documentclass{llncs}
\usepackage{algorithmic}
\usepackage[boxed]{algorithm}
\usepackage{amsmath} 
\usepackage{graphicx}
\usepackage{amsfonts}
\usepackage{booktabs}
\usepackage{caption}
\usepackage{geometry}

\usepackage{varioref} 
\usepackage{array}

\usepackage{float}
\floatstyle{boxed}
\restylefloat{figure}
\usepackage{tikz}
\usetikzlibrary{trees, shapes.geometric}
\usepackage{footnote}
\usepackage{lipsum}
\usepackage{paralist}

\newcommand{\nextitem}{\par\hspace*{\labelsep}\textbullet\hspace*{\labelsep}}

\usepackage{amsmath,amssymb}
\let\doendproof\endproof
\renewcommand\endproof{~\hfill\qed\doendproof}
\begin{document}
\title{Evaluation of DNF Formulas}
\author{Sarah R.~Allen\inst{1, 2} \and Lisa Hellerstein\inst{2} \and Devorah Kletenik\inst{2}\and Tongu\c{c} {\"U}nl{\"u}yurt\inst{3}}
\institute{Carnegie Mellon University \and Polytechnic Institute of NYU \and Sabanci University}
\date{\parbox{\linewidth}{\centering%
  \today\endgraf\bigskip
  DRAFT: \endgraf \medskip Please do not distribute}}
\maketitle

\section{Introduction}
Stochastic Boolean Function Evaluation (SBFE) is the problem of determining
the value of a given Boolean function $f$ on an unknown input $x$,
when each bit of $x_i$ of $x$ can only be determined by paying a
given associated cost $c_i$.  Further, $x$ is drawn from a
given product distribution: for each $x_i$, $Prob[x_i=1] = p_i$,
and the bits are independent.
The goal is to minimize the expected
cost of evaluation.  This problem has been studied in the Operations Research
literature, where it is known as ``sequential testing'' of Boolean functions (cf. ~\cite{unluyurtReview}).
It has been studied in learning theory in the context of
learning with attribute costs~\cite{kaplanMansour-Stoc05}.

In this paper, we study the complexity
of the SBFE problem for classes of DNF formulas.
We consider both exact and approximate versions of the problem
for subclasses of DNF, for arbitrary costs and product distributions,
and for unit costs and/or the uniform distribution.
Because of the \textsf{NP}-hardness of satisfiability, the general SBFE problem is
easily shown to be \textsf{NP}-hard for arbitrary DNF formulas
~\cite{Greiner06}.

We consider the SBFE problem for monotone $k$-DNF and $k$-term DNF
formulas.
We use a simple reduction to show that the SBFE problem for $k$-DNF
is \textsf{NP}-hard, even for $k=2$.
We present an algorithm for evaluating monotone $k$-DNF that
achieves a solution that is within a factor of
$\frac{4}{\rho^k}$ of optimal, where $\rho$ is either the minimum $p_i$ value, or the
minimum $1-p_i$ value, whichever is smaller.
We present an
algorithm for evaluating monotone $k$-term DNF with
an approximation factor of $\max\{2k, \frac{2}{\rho}(1 + \ln k)\}.$
We also prove that the SBFE problem for monotone $k$-term DNF
can be solved exactly in polynomial time for constant $k$.

Previously, Kaplan et al. gave an approximation algorithm
solving the SBFE problem for CDNF formulas (and decision trees)
for the special case of
unit costs, the uniform distribution, and monotone CDNF formulas
~\cite{kaplanMansour-Stoc05}.
CDNF formulas are formulas
consisting of a DNF formula together with an equivalent CNF formula,
so the size of the input depends both on the size of the
CNF and the size of the DNF. Having both formulas makes the evaluation problem easier.
They showed that their algorithm achieves
a solution whose cost is within an $O(\log kd)$ factor of
the expected certificate cost, where $k$ is the number
of terms of the DNF, and $d$ is the number of clauses.
The expected certificate cost is a lower bound on the cost
of the optimal solution.
Deshpande et al. subsequently
gave an algorithm solving the unrestricted
SBFE problem for CDNF formulas,
whose solution is within a factor of $O(\log kd)$
of optimal, for arbitrary costs, arbitrary
probabilities, and without the monotonicity assumption~\cite{dhkarxiv}.
Thus the Deshpande et al. result
solves a more general problem than that of Kaplan et al., but their
approximation bound is weaker because it is not in terms
of expected certificate cost.

The Kaplan et al. algorithm
uses a round robin technique that alternates between two processes, one of which attempts to achieve a 0-certificate and one which attempts to achieve a 1-certificate. The technique requires unit costs. We show how to modify the technique to
handle non-unit costs, with no change in the approximation bound.
The algorithm can also be trivially extended to
remove the uniform distribution restriction,
changing the approximation bound to $O(\frac{1}{\rho} \log kd)$.

We do not know how to
remove the assumption of Kaplan et al. that the CDNF formula
is monotone, while still achieving an approximation factor that is
within $O(\log kd)$ of the expected certificate cost.
We do show, however, that this approximation factor
is close to optimal, even for the special case they considered.
We prove that, with respect to the expected certificate cost,
the approximation factor must be at least
$\Omega((\log kd)^{\epsilon})$,
for any constant $\epsilon$ where $0 < \epsilon < 1$.

This proof also implies that
the (optimal)
average depth of a decision tree computing
a Boolean function can be exponentially
larger than the average certificate size
for that function
(i.e., the average of the minimum-size certificates
for all $2^n$ assignments).
In contrast, the depth complexity of a decision tree for a function,
(a worst-case measure) is at most
quadratic in its certificate complexity
(cf. ~\cite{BuhrmanWolf}).

\section{Stochastic Boolean Function Evaluation}
The formal definition of the Stochastic Boolean Function Evaluation (SBFE) problem is
as follows.
The input is a representation of a Boolean function $f(x_1, \ldots, x_n)$ from a fixed class of representations $C$,
a probability vector $p = (p_1, \ldots, p_n)$, where $0 < p_i < 1$,
and a real-valued cost vector $(c_1, \ldots, c_n)$, where $c_i \geq 0$.
An algorithm for this problem must
compute and output the value of $f$ on an $x \in \{0,1\}^n$, drawn
randomly from the product distribution $D_p$,
i.e., the distribution where $p_i = Prob[x_i = 1]$ and the $x_i$ are independent.
However, the algorithm is not given direct access to $x$.  Instead, it can
discover the value of any $x_i$ only by ``testing'' it, at a cost of $c_i$.
The algorithm must perform the tests sequentially,
each time choosing the next test to perform.
The algorithm can be adaptive, so
the choice of the next test can depend on the
outcomes of the previous tests.
The expected cost of the algorithm is the cost it incurs on a random $x$ from $D_p$.
(Note that since each $p_i$ is strictly between 0 and 1, the algorithm must
continue doing tests until it has obtained a 0-certificate or 1-certificate
for the function.)
The algorithm is optimal if it
has the minimum possible expected cost with respect to $D_p$.

We consider the running time of the algorithm to be the (worst-case) time it takes to
determine the single next variable to be tested, or to compute the value of $f(x)$ after
the last test result is received.
The algorithm corresponds to a Boolean decision tree (testing strategy) computing $f$, indicating
the adaptive sequence of tests.

SBFE problems arise in many
different application areas.
For example, in medical diagnosis,
the $x_i$ might
correspond to medical tests performed on a given patient,
where $f(x) = 1$ if the patient should be diagnosed as having a particular disease.
In query optimization in databases, $f$ could correspond to a Boolean query, on predicates corresponding to $x_1, \ldots, x_n$,
that has to be evaluated for every tuple in the database in order to find tuples satisfying the query~\cite{Ibaraki1984,KrishnamurthyBZ86,conf/icde/DeshpandeH08,SrivastavaMWM06}.

There are polynomial-time algorithms solving the SBFE problem exactly
for a small number of classes of Boolean formulas, including
read-once DNF formulas and $k$-of-$n$ formulas (see~\cite{unluyurtReview} for a survey
of exact algorithms).
There is a naive approximation algorithm for evaluating any function under any distribution that achieves an approximation factor of $n$:
Simply test the variables in increasing order of their costs.
This follows easily from the fact that the cost incurred
by the naive algorithm in evaluating function $f$
on an input $x$ is at most $n$ times the cost of the min-cost certificate
for $f$,  contained in $x$  (cf.~\cite{kaplanMansour-Stoc05}).

Deshpande et al. explored a generic approach to
developing approximation algorithms for SBFE problems, called the $Q$-value
approach. It involves
reducing the problem to an instance of Stochastic Submodular Set Cover
and then solving it using the Adaptive Greedy algorithm of Golovin and Krause~\cite{golovinKrause}.
They proved that the $Q$-value approach does not yield a sublinear approximation
bound for evaluating $k$-DNF formulas, even for $k = 2$.
They also developed a new algorithm for solving Stochastic
Submodular Set Cover, called Adaptive Dual Greedy,
and used it to obtain a 3-approximation algorithm solving
the SBFE problem for linear threshold formulas~\cite{dhkarxiv}.

Table~\ref{tab:results} summarizes work on the SBFE problem for classes of
DNF formulas, and for monotone versions of those classes.
The table includes both previous results and the results in this paper.
\begin{savenotes}
\begin{table}
\centering
\caption{Complexity of the SBFE Problem for DNF Formulas}
\label{tab:results}

\begin{tabular}{p{3cm}p{5cm}p{7.5cm} }
\toprule
DNF formula
& general case & monotone case\\
\midrule
\addlinespace[-.5em]
read-once DNF &
\parbox{4cm}{\begin{itemize}\item[$\bullet$]$O(n \ln n)$-time \\ algorithm~\cite{kaplanMansour-Stoc05,Greiner06}\end{itemize}}
 &
\nextitem $O(n \ln n)$-time algorithm~\cite{kaplanMansour-Stoc05,Greiner06} \\
\addlinespace[-.5em]
\midrule
\addlinespace[-.5em]
$k$-DNF & \parbox{4.5cm}
{ \begin{itemize}\item[$\bullet$] inapproximable even \\ under \emph{ud} (\S~\ref{sec:hardnessmonotone}) \end{itemize}
 } &
\parbox{7.5cm}{
\nextitem \textsf{NP}-hard, even with \emph{uc}~(\S~\ref{sec:hardnessmonotone})
\nextitem poly-time $(\frac{4}{\rho^k})$-approximation algorithm (\S~\ref{sec:monotonek})
} \\ 
\addlinespace[-.5em]
\midrule
\addlinespace[-.5em]

$k$-term DNF & \parbox{4.5cm}{\begin{itemize}[$\bullet$]\item polynomial-time \mbox{$O(k \log n)$-approximation \cite{dhkarxiv}\footnote{This follows from the fact that any DNF formula with at most $k$ terms can be expressed as a CNF formula with at most $n^k$ clauses.} }\end{itemize}}
&\parbox{7.5cm}{\begin{itemize}[$\bullet$]
\item $O(n^{2^k})$-time algorithm for general case (\S~\ref{sec:exactmonotonekterm})
\item $O(2^{2^k})$-time algorithm for \emph{uc/ud} case (\S~\ref{sec:exactmonotonekterm})
\item  polynomial-time \\
\mbox{$\max\{2k, \frac{2}{\rho}(1+\ln k)\}$-approximation}
 (\S~\ref{sec:monotonekterm})
 \end{itemize}
}
\\
\addlinespace[-.5em]
\midrule
\addlinespace[-.5em]
CDNF & \parbox{4.5cm}{\begin{itemize}[$\bullet$] \item polynomial-time \mbox{$O(\log(kd))$-approximation } (wrt~E[OPT])~\cite{dhkarxiv}\end{itemize}} &
\parbox{7.5cm}{\begin{itemize}[$\bullet$]
\item No known polynomial-time exact algorithm or \mbox{\textsf{NP}-hardness} proof
\item poly-time $O(\log(kd))$-approx. for \emph{uc} and \emph{ud} ~\cite{kaplanMansour-Stoc05}
\item poly-time $O(\log(kd))$-approx. (wrt~E[OPT])\cite{dhkarxiv} \end{itemize}} \\ 
\addlinespace[-.5em]
\midrule
\addlinespace[-.5em]
general DNF &
\parbox{4.5cm}{\begin{itemize}[$\bullet$]
\item inapproximable even under \emph{ud} (\S~\ref{sec:hardnessmonotone}) \end{itemize}} & \parbox{7.5cm}{
\begin{itemize}[$\bullet$]
\item \textsf{NP}-hard, even with \emph{uc} (\S~\ref{sec:hardnessmonotone})\item inapproximable within a factor of \\ $c \ln{n}$ for a constant $c$ (\S~\ref{sec:hardnessmonotone})
\end{itemize}}\\
\addlinespace[-.5em]
\bottomrule
\addlinespace[1em]

\end{tabular}

{\parbox{16cm}{The abbreviations \emph{uc} and \emph{ud} are used to refer to unit costs and uniform distribution, respectively. $k$ refers to the number of terms in the DNF, $d$ refers to the number of clauses in the CNF. $\rho$ is the minimum value of any $p_i$ or $1-p_i$. Citations of results from this paper are enclosed in parentheses and include the section number.
All
approximation factors are with respect to E[CERT], the expected certificate cost, except for the CDNF bound of~\cite{dhkarxiv}. That bound is with respect to
E[OPT], the expected cost of the optimal strategy, which is lower bounded by E[CERT].}
}
\end{table}

\section{Preliminaries}
\subsection{Definitions}

A \emph{literal} is a variable or its negation.
A \emph{term} is a possibly empty conjunction ($\wedge$) of literals.
If the term is empty, all assignments satisfy it. A clause is a possibly empty disjunction ($\vee$) of literals. If the clause is empty, no assignments satisfy it. The {\em size} of a term or clause is the number of literals in it.
%

A \emph{DNF} (disjunctive normal form) formula is either the
constant 0, the constant 1, or
a formula of the form
$t_1 \vee \dots \vee t_k$, where $k \geq 1$ and each $t_i$ is a term. Likewise, a \emph{CNF} (conjunctive normal form) formula is either the constant 0, the constant 1, or a formula of the form $c_1 \wedge \dots \wedge c_k$, where each $c_i$ is a clause.

A $k$-term DNF is a DNF formula consisting of at most $k$ terms.
A $k$-DNF is a DNF formula where each term has size at most $k$.
The {\em size} of a DNF (CNF) formula is the number of its terms (clauses); if it is
the constant 0 or 1, its size is 1.
A DNF formula is \emph{monotone} if it contains no negations. A \emph{read-once} DNF formula is a DNF formula where each variable appears at most once.


Given a Boolean function $f:\{0,1\}^n \rightarrow \{0,1\}$,
a partial assignment $b \in \{0,1,*\}^n$
is a \emph{0-certificate (1-certificate)} of $f$ if
 $f(a) = 0$ ($f(a)=1$) for all $a$ such that $a_i = b_i$ for all $b_i \neq *$.
It is a certificate for $f$ if it is either a 0-certificate or a 1-certificate.
Given a cost vector $c = (c_1, \ldots, c_n)$,
the cost of a certificate $b$ is $\sum_{j:b_j \neq *} c_j$.
We say that input $x$ {\em contains} certificate $b$ if
$x_i = b_i$ for all $i \neq *$.
The \emph{variables in a certificate} $b$
are the $x_i$ such that $b_i \neq *$.
If $x$ contains $b$ and $S$ is a superset of the variables in $b$,
then we say that $S$ contains $b$.

The expected certificate cost of a function $f$,
with respect to cost vector $c$ and probability vector $p$,
is $E[CERT(f,x)]$, where the expectation is with respect to
$x$ drawn from product distribution $D_p$, and $CERT(f,x)$
is the minimum cost of a certificate $b$ of $f$ contained in $x$.

Given a Boolean function $f$,
let $E_f[OPT]$ denote the
minimum expected cost of any algorithm solving the SBFE
for $f$, in the unit-cost, uniform distribution case.
Let $E_f[CERT]$ denote the expected certificate cost,
in the unit cost, uniform distribution case.


The {\em set covering problem} is as follows: Given a ground set $A = \{e_1, \ldots, e_m\}$ of
elements, a set ${\cal S} = \{S_1, \ldots, S_n\}$ of subsets of $A$, and a positive integer $k$,
does there exist ${\cal S}' \subseteq {\cal S}$ such that $\bigcup_{S_i \in {\cal S}'} = {\cal S}$ and $|{\cal S}'| \leq k$?
Each set $S_i \in {\cal S}$ is said to cover the elements it contains.  Thus the set covering
problem asks whether
$A$ has a ``cover'' of size at most $k$.


\end{savenotes}
\section{Hardness of the SBFE problem for monotone DNF}
\label{sec:hardnessmonotone}

Before presenting approximation algorithms solving the SBFE problem
for classes of monotone DNF, we begin by discussing the hardness
of the exact problem.

Greiner et al.~\cite{Greiner06} showed that the SBFE problem for
CNF formulas is \textsf{NP}-hard, as follows.
If a CNF formula is unsatisfiable, then no tests
are necessary to determine its value
on an assignment $x$.
If there were a polynomial-time algorithm solving the SBFE problem
for CNF formulas, we could use it to solve SAT:
given CNF Formula $\phi$, we could  run the SBFE algorithm
on $\phi$ (with arbitrary $p$ and $c$),
and just observe whether the algorithm begins
by choosing a variable to test, or
whether it immediately outputs 0 as the value of the formula.
Thus the SBFE problem on CNF formulas is \textsf{NP}-hard, and
by duality, the same is true for DNF formulas.

Moreover, if $\textsf{P} \neq \textsf{NP}$, we cannot approximate
the SBFE problem for DNF within any factor $\rho > 1$.
If a $\rho$-approximation algorithm existed, then on a tautological DNF $\phi$,
the algorithm would have to immediately output 1 as the value
of $\phi$, because $\rho \times 0 = 0$.  On non-tautological $\phi$,
the algorithm would instead have to specify a variable to test.

The SBFE problem for DNF is still
$\textsf{NP}$-hard even when the DNF is monotone.
To show this, we use an approach used by
Cox~\cite{heuristicLeastCostCox} in proving $\textsf{NP}$-hardness of linear threshold evaluation.
Intuitively, in an instance of SBFE with unit costs
if the probabilities $p_i$ are very close to 0 (or 1), then
the expected cost of evaluation is dominated
by the cost of evaluating the given function $f$ on a specific input $x^*$.
That cost is minimized by testing
only the variables in a minimum-cost certificate for $f$ on $x^*$.
The idea, then, is to show hardness of the SBFE problem for a class of formulas $C$
by reducing an $\textsf{NP}$-hard problem to the problem of finding,
given $f \in C$ and a particular input $x^*$,
a smallest size certificate of $f$ contained in $x^*$.
Cox reduced from Knapsack, and here we reduce from Vertex-Cover.
The following lemma is implicit in the proof of Lemma 1 of Cox:
\begin{lemma} \label{lem:allzerosoptimal}
 Let $T$ be a Boolean decision tree computing Boolean function $f(x_1, \ldots, x_n)$.
For $i \in \{1,\ldots, n\}$, let $q_i = q = \max\{ \frac{n^2-.5}{n^2}, (\frac{d+.5}{d+1})^{1/n}\}$
and let $p_i = 1-q_i$.
Let $p = (p_1, \ldots, p_n)$, let $c = (1, \ldots, 1)$ be the vector of unit costs,
and let $0^n$ denote the all 0's assignment.
If with respect to $D_p$ and $c$,
$T$ has minimum expected evaluation cost over all decision trees computing $\phi$,
then the variables tested along the path
corresponding to $0^n$ in $T$ are precisely those set to $0$ in a min-cost certificate
for $f$ contained in $0^n$.
\end{lemma}

\begin{theorem} If $\textsf{P} \neq \textsf{NP}$, there is no polynomial time  algorithm solving the SBFE problem for monotone DNF.\ This holds even with unit costs, and even for $k$-DNF where $k \geq 2$. Also, if
$\textsf{P} \neq \textsf{NP}$,
the SBFE problem for monotone DNF, even with unit costs, cannot be approximated to within a factor of less than $c \ln n$, for some constant $c$.\end{theorem}

\begin{proof}
Suppose there is a polynomial-time algorithm ALG for the SBFE problem for monotone 2-DNF,
with unit costs and arbitrary probabilities.
We show this algorithm could be used to solve the Vertex Cover problem:
Given a graph $G(V,E)$, find a minimum-size vertex cover for $G$,
i.e., a minimum-size set of vertices $V' \subseteq V$
such that for each edge $(v_j, v_k) \in E$, $\{v_j,v_k\} \bigcap V' \neq \emptyset$.

The reduction is as follows.
Given graph $G(V, E)$, construct a monotone 2-DNF formula $\phi$ whose variables $x_j$
correspond to the vertices $v_j \in V$,
and whose terms $x_j x_k$ correspond to the edges $e = (v_j, v_k)$ in $E$.
Consider the all 0's assignment $0^n$.
Since a 0-certificate for $\phi$
must set each term of $\phi$ to 0,
any min-cost certificate for $\phi$ contained in $0^n$
must also be a minimum-size vertex cover for $G$.
Thus by the previous lemma,
one can find a minimum-size vertex cover for $G$ by
using ALG to evaluate $\phi$ on input $0^n$,
with unit costs and the probabilities $p_i$
given in Lemma~\ref{lem:allzerosoptimal}, and observing which variables are tested.

A more general version of this reduction can be used to reduce
the general Set Cover problem to the SBFE problem for monotone DNF
(with terms of arbitrary length).  The non-approximability bound in the theorem
then follows from the $c \ln n$ inapproximability result for Set Cover~\cite{feige}.
\end{proof}

Given the difficulty of exactly solving the SBFE problem for monotone DNF formulas, we now consider approximation algorithms.

%

\section{Approximation algorithms for the evaluation of monotone $k$-DNF and $k$-term DNF} \label{sec:dnfeval}

\subsection{Monotone $k$-DNF formulas}
\label{sec:monotonek}

In this section, we will present a polynomial time algorithm for evaluating monotone $k$-DNF formulas. 
To evaluate $f$ we will alternate
between two algorithms, Alg0 and Alg1, each of which performs tests
on the variables $x_i$.
Alg0 tries to find a min-cost 0-certificate
for $f$, and Alg1 tries to find a min-cost 1-certificate for $f$.
As soon as one of these algorithms succeeds in finding a certificate,
we know the value of $f(x)$, and can output it.

This basic approach was used previously by
Kaplan et al.~\cite{kaplanMansour-Stoc05} in their algorithm
for evaluating monotone CDNF formulas in the unit cost, uniform distribution
case.  They used a standard greedy set-cover algorithm for both Alg0
and Alg1, with a strict round-robin policy that alternated between doing one
test of Alg0 and one test of Alg1.
Our algorithm uses a dual greedy set-cover algorithm for Alg0 and a
different, simple algorithm for Alg1.
The strict round-robin policy used by Kaplan et al. is only suitable
for unit costs, and our algorithm has to handle arbitrary costs.
Our algorithm uses a modified round-robin protocol instead.  We
begin by presenting that protocol.

Although we will use the protocol with a particular Alg0 and Alg1,
it works for any Alg0 and Alg1 that ``try'' to find 0-certificates
and 1-certificates respectively.
In the case of Alg0, this means that
Alg0 will succeed in outputing a 0-certificate of $f$
contained in $x$ if $f(x) = 0$, and
will eventually terminate and report failure otherwise.
Similarly, Alg1 will output a 1-certificate contained in $x$ if $f(x) = 1$, and will report failure otherwise.

The modified round-robin protocol works as follows.
It maintains two values: $K_0$ and $K_1$, where
$K_0$ is the cumulative cost of all tests performed
so far in Alg0, and $K_1$ is the cumulative cost of all tests
performed so far in Alg1.
At each step of the protocol, each of Alg0 and Alg1 independently determines a test to be performed next and the protocol chooses one of them.
(Initially, the two tests are the first tests of Alg0 and Alg1 respectively.)
Let $C_0$ and $C_1$ denote the respective costs of these tests.
Let ${x_j}_1$ denote the next test of Alg1 and let ${x_j}_0$ denote the next test of Alg0.
To choose which test to perform, the protocol uses the following
rule: {\em if
$K_0 + C_0 \leq K_1 + C_1$ it performs test ${x_j}_0$,
otherwise it performs test ${x_j}_1$.}

The result of the test is given to the algorithm to which it belongs,
and that algorithm continues until it either
(1) computes a new next test,
(2) terminates successfully and outputs a certificate,
or (3) terminates by reporting failure.  In the first case,
the protocol again chooses between the next test of Alg0 and Alg1,
using the rule above.
In the second, the protocol terminates because one of the
algorithms has output a certificate.
In the third,
the protocol runs the other algorithm (the one that did not terminate) until completion,
performing all of its remaining tests.
That algorithm is guaranteed to
output a
certificate, because if $x$ doesn't have a 0-certificate for $f$, it
must have a 1-certificate, and vice-versa.

Note that it would be possible for the above protocol to share information
between Alg0 and Alg1, so that if $x_i$ was tested by Alg0, Alg1 would not need
to retest $x_i$.  However, to simplify the analysis, we do not have the protocol
do such sharing.
We now show that the following invariant holds
holds at the end of each step of the protocol, provided that
neither Alg0 nor Alg1 terminated in that iteration.

\begin{lemma} \label{invariant}
At the end of each step of the above modified round-robin protocol,
if ${x_j}_1$ was tested in that step,
then $K_1 - {c_j}_0 \leq K_0 \leq K_1$. Otherwise,
if ${x_j}_0$ was tested, then $K_0 - {c_j}_1 \leq K_1 \leq K_0$
at the end of the step. \end{lemma}

\begin{proof}
The invariant clearly holds
after the first step. Suppose it is true at the end of the $k$\textsuperscript{th} step,
and without loss of generality assume that ${x_j}_1$ was tested during that step.
Thus $K_1 - {c_j}_0 \leq K_0 \leq K_1$ at the end of the $k$\textsuperscript{th} step.

Consider the $k+1$\textsuperscript{st} step.  Note that the value of ${x_j}_0$
is the same in this step as in the previous one,
because in the previous step, we did not execute the next
step of Alg0. There are 2 cases, depending on which if-condition is satisfied when the rule is applied in this step,
$K_0 + {c_j}_0 \leq K_1 + {c_j}_1$ or $K_0 + {c_j}_0 > K_1 + {c_j}_1$.

\textbf{Case 1:} $K_0 + {c_j}_0 \leq K_1 + {c_j}_1$ is satisfied.

Then ${x_j}_0$ is tested in this step and $K_0$ increases by $c_0$. We show that
$K_0 - {c_j}_1 < K_1$ and
$K_1 \leq K_0$ at the end of the step, which is what we need. At the start of the step, $K_0 + {c_j}_0 \leq K_1 + {c_j}_1$
and  at the end, $K_0$ is augmented by ${c_j}_0$, so
$K_0 \leq K_1 + {c_j}_1$. Consequently,
$K_0 - {c_j}_1 \leq K_1$.
Further, by assumption,
$K_1 - {c_j}_0 \leq K_0$ at the start of the step,
and hence at the end, $K_1 \leq K_0$.

\textbf{Case 2:} $K_0 + {c_j}_0 > K_1 + {c_j}_1$ is satisfied
[and by assumption $K_1 - {c_j}_0 \leq K_0 \leq K_1$ at the start]

Then ${x_j}_1$ is tested in this step, and $K_1$ increases by $c_1$.
We show that
$K_1 - {c_j}_0  \leq K_0$ and $K_0 \leq K_1$ at the end of the step. By the condition in the case, $K_0 +{c_j}_0 > K_1 + {c_j}_1$ at the
start of the step, so
at the end, $K_0 + c_{j_0} > K_1$,
and hence $K_1 - {c_j}_0 < K_0$. Further, by assumption
$K_0 \leq K_1$ at the start, and since only $K_1$ was increased,
this also holds at the end.
\end{proof}

We can now prove the following lemma:
\begin{lemma}
\label{lem:costofalgorithm}
If $f(x) = 1$, then at the end of the modified round-robin protocol,
$K_1 \geq K_0$. The lemma holds true symmetrically if $f(x)=0$. \end{lemma}

\begin{proof}
There are two ways for the protocol to terminate.
Either Alg0 or Alg1 is detected to have succeeded at the
start of the repeat loop,
or within the loop,
one fails and the other is run to successful termination.

Suppose the former, and without loss of generality suppose it is Alg0 that
succeeded.  It follows that it was ${x_j}_0$ that was tested
at the end of the previous step (unless
this is the first step, which would be an easy case),
because otherwise, the success of Alg0
would have been detected in an earlier step.

Thus at the end of the last step, by Lemma \ref{invariant},
$K_1 \leq K_0$.

Suppose instead that one algorithm
fails, and without loss of generality, suppose it was Alg0, and thus we ran Alg1 to termination.  Since Alg0 did not fail in a prior
step,
it follows that in the previous step,
${x_j}_0$ was tested  (unless this is the first step, which
would be an easy case).  Thus at the end of the previous step,
by the invariant, $K_0 - {c_j}_1 \leq K_1$
and so $K_0 \leq K_1 + {c_j}_1$.
We have to run at least one step of Alg1 when we run it to
termination.  Thus running Alg1 to termination augments $K_1$ by ${c_j}_1$, and so
at the end of the algorithm, we have $K_0 \leq K_1$.
\end{proof}

We now describe the particular Alg0 and Alg1 that we use in our algorithm
for evaluating monotone $k$-DNF.
We describe Alg0 first.
Since $f$ is a monotone function, the variables in
any 0-certificate for $f$ must all be set to 0.
Consider an assignment $x \in \{0,1\}^n$ such that $f(x) = 0$.
Let $Z = \{x_i | x_i = 0\}$.
Finding a min-cost 0-certificate for $f$ contained in $x$
is equivalent to solving the set-cover instance where the
elements to be covered are the terms
$t_1, \ldots, t_m$, and for each $x_i \in Z$,
there is a corresponding subset
$\{t_j | x_i \in t_j\}$ .

Suppose $f(x) = 0$.
If Alg0 was given both $Z$ and $f$ as input, it could find an
approximate solution to this set cover instance using
Hochbaum's  Dual Greedy algorithm
for (weighted) set cover~\cite{hochbaum}.
This algorithm selects items to place in the cover, one by one,
based on a certain greedy choice rule.

Alg0 is not given $Z$, however.  It can only discover
the values of variables $x_i$ by testing them.
We get around this as follows.  Alg0 begins running Hochbaum's
algorithm, using the assumption that all variables are in $Z$.
Each time that algorithm
chooses a variable $x_i$ to place in the cover,
Alg0 tests the variable $x_i$.
If the test reveals that $x_i = 0$, Alg0 continues directly to the next step
of Hochbaum's algorithm.
If, however, the test reveals that $x_i = 1$, it
removes the $x_i$ from consideration, and uses the greedy choice
rule to choose the best variable from the remaining variables.
The variables that are placed in the
cover by Alg0 in this case are precisely those
that would have been placed in the
cover if we had run Hochbaum's algorithm with
$Z$ as input.

Hochbaum's algorithm is guaranteed to construct a cover whose total cost is
within a factor of $\alpha$ of the optimal cover, where $\alpha$ is
the maximum number of subsets in which any ground element appears.
Since each term $t_j$ can contain a maximum of $k$ literals, each term can be covered at most $k$ times.
It follows that when $f(x) = 0$, Alg0 outputs a certificate that is within a factor
of at most $k$ of the minimum cost certificate of $f$ contained in $x$.

If $f(x) = 1$, Alg0 will eventually test all elements without
having constructed a cover,
at which point it will terminate and report failure.

We now describe Alg1.
Alg1 begins by evaluating
the min-cost term $t$ of $f$, where the cost of a term
is the sum of the costs of the variables in it.
(In the unit-cost case, this is the shortest term.
If there is a tie for the min-cost term, Alg1
breaks the tie in some suitable way, e.g., by the lexicographic ordering
of the terms.)
The evaluation is done by testing the variables of
$t$ one by one in increasing cost order
until a variable is found to equal 0,
or all variables have been found to equal 1.
(For variables $x_i$ with equal cost, Alg1 breaks ties in some suitable way, e.g.,
in increasing
order of their indices $i$.)
In the latter case,
Alg1 terminates and
outputs the certificate setting the variables in the
term to 1.

Otherwise, for each tested variable in $t$, Alg1 replaces all
occurrences of that variable in $f$ with its tested value.
It then simplifies the formula
(deleting terms with 0's and deleting 1's from terms,
and optionally making the resulting formula minimal).
Let $f'$ denote the simplified formula.
Because $t$ was not satisfied,
$f'$ does not contain any
satisfied terms.
If $f'$ is identically 0, $x$ does not contain a 1-certificate
and Alg1 terminates unsuccessfully.
Otherwise, Alg1
proceeds recursively on the simplified formula, which contains only
untested variables.

Having presented our Alg0 and Alg1, we are ready to prove
the main theorem of this section.

\begin{theorem}
\label{thm:kdnfcosts}
The evaluation problem for monotone $k$-DNF
can be solved
by a polynomial-time approximation algorithm computing a strategy
that is within a factor of $\frac{4}{\rho^k}$ of the expected certificate cost.
\end{theorem}

\begin{proof}
Let $f$ be the input monotone $k$-DNF, defined on $x \in \{0,1\}^n$.
We will also use $f$ to denote the function computed by this formula.

Let Alg be the algorithm
for evaluating $f$ that
alternates between the Alg0 and Alg1 algorithms just described, using the modified
round-robin protocol.

Let $S_1 = \{x | f(x) = 1\}$ and $S_0 = \{x | f(x) = 0\}$.
Let $E_f[cost]$ denote the expected cost incurred by the
round-robin algorithm in evaluating $f$ on random $x$.
Let $cost_f(x)$ denote the cost incurred by running the algorithm on $x$.
Let $Pr[x]$ denote the probability that $x =1$ with respect to the product distribution
$D_p$.
Thus $E_f[cost]$ is equal to
$\sum_{x \in S_1} Pr[x]cost_f(x) + \sum_{x \in S_0} Pr[x]cost_f(x)$.
Similarly, let $cert_f(x)$ denote the cost of the minimum cost certificate
of $f$ contained in $x$.
We need to show that the ratio between $E_f[cost]$ and $E_f[cert]$ is
at most $\frac{4}{\rho^k}$.

We consider first the costs incurred by Alg0 on inputs $x \in S_0$.
Following the approach of Kaplan et al.,
we divide the tests performed by Alg0 into two categories,
which we call useful and useless, and amortize the cost of the
useless tests by charging them to the useful tests.
More particularly, we say that
a test on variable $x_i$ is useful to Alg0 if $x_i = 0$
($x_i$ is added to the 0-certificate in this case)
and useless if $x_i = 1$.
The number of useful tests on $x$ is equal to the size of the
certificate output by Alg0, and thus the total cost of the
useful tests Alg0 performs on $x$ is at most $k(cert_f(x))$.

Let $cost_f^0(x)$ denote the cost incurred by Alg0 alone when running
Alg to evaluate $f$ on $x$, and
let $cost_f^1(x)$ denote the cost incurred by Alg1 alone.
Suppose Alg0 performs a useless test on an $x \in S_0$, finding that $x_i=1$.
Let $x'$ be the assignment produced from $x$ by setting $x_i$ to $0$.
Because $f(x)=0$ and $f$ is monotone, $f(x') = 0$ too.
Because $x$ and $x'$ differ in only one bit, if Alg0 tests $x_i$
on assignment $x$, it will test $x_i$ on $x'$, and that
test will be useful.
Thus each useless test performed by Alg0 on $x\in S_0$
corresponds to a distinct useful test
performed on an $x' \in S_0$.
When $x_i$ is tested, the probability that it is 1 is $p_i$,
and the probability that it is 0 is $1-p_i$.
Each useless test contributes $c_ip_i$ to the expected cost,
whereas each useful test contributes $(1-p_i)c_i$.
If we multiply the contribution of the useful test by $1/(1-p_i)$,
we get the contribution of both a useful and a useless test, namely $c_i$.
To charge the cost of a useless test to its corresponding useful test,
we can therefore multiply the cost of the useful test by $1/(1-p_i)$
(so that if, for example, $p_i = 1/2$, we charge double for the
useful test).
Because $1/(1-p_i) \leq \frac{1}{\rho}$ for all $i$,
it follows that
$\sum_{x \in S_0} Pr[x] cost_f^0(x) \leq \frac{1}{\rho} \sum_{x \in S_0} Pr[x] (k(cert_f(x)))$.
Hence,
$$\frac{\sum_{x \in S_0} Pr[x]cost_f^0(x)}{\sum_{x \in S_0} Pr[x]cert_f(x)} \leq \frac{k}{\rho}.$$

We will now show, by induction on the number of terms
of $f$, that
$E[cost_f^1(x)]/E[cert_f(x)] \leq \frac{1}{\rho^k}$.

If $f$ has only one term,
it has at most $k$ variables.
In this case, Alg1 is just using the na\"{i}ve algorithm which tests
the variables in increasing cost order until the function
value is determined.
Since the cost of using the na\"{i}ve algorithm
on $x$ in this case
is at most $k$ times $cert_f(x)$,
$\rho \leq 1/2$, and $k \leq 1/2^k$ for all $k \geq 1$, it follows that
$E[cost_f^1(x)]/E[cert_f(x)] \leq k \leq \frac{1}{2^k} \leq \frac{1}{\rho^k}$.
Thus we have the base case.

Assume for the purpose of induction that
$E[cost_f^1(x)]/E[cert_f(x)] \leq \frac{1}{\rho^k}$ holds for $f$ having
at most $m$ terms.
Suppose $f$ has $m+1$ terms.
Let $t$ denote the min-cost term.
Let $C$ denote the cost of $t$, and $k'$ the number of variables in $t$,
so $k' \leq k$.
If $x$ does not satisfy term $t$, then after
Alg1 evaluates term $t$ on $x$,
the results of the tests performed in the evaluation
correspond to a partial assignment $a$
to the variables in $t$.
More particularly, if Alg1 tested exactly
$z$ variables of $t$, the test results
correspond to the
partial assignment $a$ setting the $z-1$ cheapest variables
of $t$ to 1 and the $z$\textsuperscript{th} to 0, leaving all other
variables in $t$ unassigned.
There are thus $k'$ possible values for $a$.
Let $T$ denote this set of partial assignments $a$.

For $a \in T$,
let $f[a]$ denote the formula obtained from $f$ by replacing
any occurrences of variables in $t$ by their assigned values in $a$
(if a variable in $t$ is not assigned in $a$, then occurrences
of those variables are left unchanged).
Let $T_0 = \{ a \in T|$ $f[a]$ is identically 0 $\}$,
and let $T_* = T - T_0$.
For any $x$, the cost incurred by Alg1 in evaluating $t$ on $x$
is at most $C$.
For $x \in T_0$, Alg1 only evaluates $t$, so its total cost on
$x$ is at most $C$.
Let $Pr[a]$ denote the joint probability of obtaining the
observed values of those variables tested in $t$. More formally, if $W$ is the set of variables tested in $t$, $Pr[a] = \prod_{i: x_i \in W \wedge x_i = 1}p_i \prod_{i:x_i \in W \wedge x_i = 0} (1-p_i)$.
We thus have the following recursive expression:
$$E[cost_f^1(x)] \leq C + \sum_{a \in T_*} Pr[a] E[cost_{f[a]}^1(x[a])]$$
where $x[a]$ is a random assignment to the variables of
$f$ not assigned values in $a$, chosen independently according to the relevant
parameters
of $D_p$.

For any $x$ satisfying $t$, since $t$ is min-cost and $f$ is monotone,
$cert_f(x) = C$.
Let $x \in T_*$, and let $a_x$ be the partial assignment
representing the results of the tests Alg1 performed in evaluating $t$
on $x$.
Let $\hat{x}$ be the restriction of $x$ to the variables of $f$ not assigned values
by $a_x$.
Any certificate for $f$ that is contained in $x$ can be converted into
a certificate
for $f[a_x]$, contained in $\hat{x}$, by simply removing the variables
assigned values by $a_x$.
It follows that $cert_f(x) \geq cert_{f[a_x]}(\hat{x})$.

Since $k' \leq k$, the probability that $x$ satisfies the first
term is at least $\rho^k$. By ignoring the $x \in T_0$ we get
$$E_f[cert] \geq \rho^kC + \sum_{a \in T_*}Pr[a]E[cert_{f[a]}(x[a])]$$

The ratio between the first term in the expression bounding $E[cost_f^1(x)]$, to the first term in the expession bounding $E_f[cert]$, is equal
to $\frac{1}{\rho^k}$.  By induction, for each $a \in T_*$,
$E[cost_{f[a]}^1(x[a])]/E[cert_{f[a]}(x[a])] \leq \frac{1}{\rho^k}$.
Thus
$$E_f^1[cost]/E_f[cert] \leq \frac{1}{\rho^k}$$

Clearly, $E_f[cost] = E_f^1[cost] + E_f^0[cost] = E_f^1[cost] + \sum_{x \in S_1}Pr[x]cost_f^0(x) + \sum_{x \in S_0} Pr[x] cost_f^0(x)$.
By Lemma~\ref{lem:costofalgorithm},
the cost incurred by Alg on any $x \in S_1$
is at most twice the cost incurred by Alg1 alone on that $x$.
Thus $E_f[cost] \leq 2E_f^1[cost]  + \sum_{x \in S_0} Pr[x] cost_f^0(x)$.
Further,
$E_f[cert] \geq \frac{1}{2}(E_f[cert] + \sum_{x \in S_0} Pr[x]cert_f(x))$ because
$x \in S_0$ contributes $Pr[x]cert_f(x)$ to both $E_f[cert]$ and
to the summation over $S_0$.

It follows from the above that $E_f[cost]/E_f[cert]$ is at most
$\max\{\frac{4}{\rho^k}, \frac{2k}{\rho}\} = \frac{4}{\rho^k}$, since
$k \geq 1$.
\end{proof}

\subsection{Monotone $k$-term DNF formulas}
\label{sec:monotonekterm}
We can use techniques from the previous subsection to obtain results for the class of monotone $k$-term DNF formulas as well. In Section \ref{sec:exactmonotonekterm}, we will present an exact algorithm whose running time is exponential in $k$. Here we present an approximation algorithm that runs in time polynomial in $n$, with no dependence on $k$.

\begin{theorem}
\label{thm:ktermdnf}

The evaluation problem for monotone $k$-term DNF can be solved by a polynomial-time approximation algorithm computing a strategy that is within a factor of  $\max\{2k, \frac{2}{\rho}(1+ \ln k)\}$  of the minimum-cost certificate.
\end{theorem}

\begin{proof}
Let $f$ be the input monotone $k$-term DNF, defined on $x \in \{0,1\}^n$.

Just as in the proof of Theorem \ref{thm:kdnfcosts}, we will utilize a modified round robin protocol that alternates between one algorithm for finding a 0-certificate (Alg0) and one for finding a 1-certificate (Alg1). Again, let $S_1 = \{x | f(x) = 1\}$ and $S_0 = \{x | f(x) = 0 \}$.

However, in this case Alg0 will use Greedy, Chv\'{a}tal's well-known greedy algorithm for
weighted set cover~\cite{Chvatalsetcoverharmonic}, instead of the Dual Greedy algorithm of Hochbaum.
The standard greedy algorithm simply  maximizes, at each iteration,
``bang for the buck'' by selecting the subset that covers the largest number
of uncovered elements relative to the cost of selecting that subset.
Greedy yields a $H(m)$ approximation, where $m$ is the number of ground
elements in the set cover instance and $H(m)$ is the $m$\textsuperscript{th} harmonic number,
which is upper bounded by $1 + \ln m$. Once again, we will view the terms
as ground elements and the variables that evaluate to 0 as the subsets.
Since $f$ has at most $k$ terms, there are at most $k$ ground elements.
On any $x \in S_0$,
Greedy will yield a certificate that is within a factor of $1 + \ln k$
of the min-cost 0-certificate $cert_f(x)$, and
thus the cost incurred by the useful tests on $x$ (tests on $x_i$ where $x_i = 0$)
is at most $cert_f(x)(1+\ln k)$. By multiplying by $1/\rho$
the charge to the variables that evaluate to 0, to account for the useless tests,
we get
that the cost incurred by Alg0 on $x$, for $x \in S_0$, is
at most
$\frac{1}{\rho}cert_f(x)(1+ \ln k)$.

Alg1 in this case simply evaluates
$f$ term by term, each time choosing the remaining term of minimum cost
and evaluating all of the variables in it.
Without loss of generality, let $t_1$ be the first
(cheapest) term evaluated by Alg1, and $t_i$ be the $i$th term evaluated.
Suppose $x  \in S_1$.  If $x$ falsifies terms $t_1$ through $t_{i-1}$
and then satisfies $t_i$, $cert_f(x)$ is precisely the
cost of $t_i$, and Alg1 terminates after evaluating $t_i$.
Since none of the costs of the first $i-1$ terms exceeds the cost of $t_i$,
the total cost of evaluating $f$
is at most $k$ times the cost of $t_i$. Hence, Alg1 incurs a cost of at most $k(cert_f(x))$.

By executing the two algorithms according to
the modified round robin protocol, we can solve the problem of evaluating
monotone $k$-term DNF with cost no greater than double the cost incurred by Alg1, when $x \in S_1$,
and no more than double the cost incurred by Alg0, when $x \in S_0$.
Hence the total cost of the algorithm is within a
factor of $\max\{2k, \frac{2}{\rho}(1+ \ln k)\}$ of the cost of the min-cost certificate for $x$.
\end{proof}

We now prove that the problem of exactly evaluating monotone
$k$-term DNF can be solved in polynomial time for constant $k$.

\section{Exact learning of monotone $k$-term DNF}
\label{sec:exactmonotonekterm}

In this section, we provide an exact algorithm for evaluating k-term DNF formulas in polynomial time for constant $k$. First, we will adapt results from
Greiner et al.~\cite{Greiner06} to show some properties of optimal strategies for monotone DNF formulas. Then we will use these properties to compute an optimal strategy monotone k-term DNF formulas.
Greiner et al. \cite{Greiner06} consider evaluating \textit{read-once} formulas with the minimum expected cost. Each \textit{read-once} formula can be described by a rooted and-or tree where each leaf node is labeled with a test and each internal node is labeled as either an or-node or an and-node. The simplest \emph{read-once} formulas are the simple AND and OR functions, where the depth of the and-or tree is 1.
Other \textit{read-once} formulas can be obtained by
taking the
AND or OR of other \textit{read-once} formulas
over disjoint sets of variables. In the and-or tree, an internal node whose children include at least one leaf is called a \textit{leaf-parent}, leaves with the same parent are called \textit{leaf-siblings} (or \textit{siblings}) and the set of all children of a \textit{leaf-parent} is called a \textit{sibling class}. Intuitively, the siblings have the same effect on the value of the \textit{read-once} formula. The ratio of a variable $x_i$ is defined to be $R(i)=\frac{p_i}{c_i}$. Further, tests $x_1$ and $x_2$ are \textit{R-equivalent} if they are \textit{leaf-siblings} and $R(x_1)=R(x_2)$. An \textit{R-class} is an equivalence class with respect to the relation of being \textit{R-equivalent}. Greiner et al. show that, for any and-or tree, (WLOG they assume that leaf-parents are OR nodes), there is an optimal strategy $S$ that satisfies the following conditions:
\begin{itemize}
\item[(a)] For any sibling tests $x$ and $y$ such that $R(y)>R(x)$, $x$ is not performed before $y$ on any root-to leaf path of $S$.
\item[(b)] For any \textit{R-class} $W$, $S$ is contiguous with respect to $W$.
\end{itemize}

We observe that by redefining siblings and sibling classes, corresponding properties hold for general monotone DNF formulas.
Let us define a maximal subset of the variables that appear in exactly the same set of terms as a \textit{sibling class} in a DNF formula.
All the other definitions can easily be adapted accordingly. In this case, the ratio of a variable $x_i$ is $R(i)=\frac{q_i}{c_i}$. For instance, all variables are siblings for an AND function, whereas no two variables are siblings in an OR function.

It is possible to adapt the proof of Theorem 20 in \cite{Greiner06} to apply
to monotone DNF formulas. All the steps of the proof can be adapted in this context, using the new definitions of siblings and the ratio of a variable.
\begin{theorem}
\label{thm:contiguity}
For any monotone DNF, there exists an optimal testing strategy $S$ that satisfies conditions (a) and (b) stated above.

\end{theorem}

In other words, there exists an optimal strategy such that on any path from the root to the leaf, sibling tests appear in non-decreasing order of their ratios. Further, for this strategy, sibling tests with the same ratio (\textit{R-class}) appear one after another on any path from the root to the leaf. By a duality argument, a similar result holds for monotone CNFs by defining the ratio of a variable as $\frac{c_i}{p_i}$ and \textit{sibling class} as a set of variables that appear in exactly the same set of clauses. For a $k$-term monotone DNF there are at most $2^k-1$ sibling classes, since each sibling class corresponds to a non-empty subset of the terms of the monotone DNF formula. Next, we provide a dynamic programming based method to find an optimal strategy.

\begin{theorem}
The evaluation problem for monotone $k$-term DNF formula $\phi$ over a product distribution on input $x$ and with arbitrary costs can be solved exactly in polynomial time for constant $k$.
\end{theorem}

\begin{proof}
We will use a dynamic programming method similar to that used in \cite{guijarroRaghavantruthtable} for building decision trees for functions defined by truth tables. We use notation consistent with that paper.

Let $f$ be the function that is defined by $\phi.$
We will construct a table $P$ indexed by partial assignments to the $t = 2^k$ sibling classes. By Theorem \ref{thm:contiguity}, there is an optimal evaluation order of the variables within each sibling class. Let $1 \leq j \leq t$ index the sibling classes in arbitrary order. For each sibling class $s^j$, let us rename the variables contained in it $x^j_i$, where $1 \leq i \leq \ell^j$ refers to the position of the variable in the testing order according to their ratios $R(i) =\frac{q_i}{c_i}$, and where $\ell^j$ refers to the number of variables within the class $s^j$. Hence, for each class we will have $\ell^j + 2$ states in $P$: not evaluated, variable $x^j_1$ evaluated to 1, variable $x^j_2$ evaluated to 1, \ldots variable $x^j_{\ell^j}$ evaluated to 1, any variable evaluated to 0. (Due to monotonicity, the evaluation of any variable to 0 ends the evaluation of the entire class.) Given the optimal ordering, the knowledge of which variable $i$ of a sibling class was evaluated last is sufficient to determine which variable $i+1$ should be evaluated next within that class. Given a partial assignment $\alpha$ that is being evaluated under an optimal testing strategy, let $s^j_i$ denote the variable $x^j_i$ that will be evaluated next for each class $s^j$ (that is, that the values of variables $x^j_h$ for all $h < i$ have already been revealed).

At each position $\alpha$ in the table, we will place the decision tree with the minimum expected cost that computes the function $f_\alpha$, where $f_\alpha$ is the function $f$ defined by $\phi$ projected over the partial assignment $\alpha$. Then, once the full table $P$ has been constructed, $P[*^n]$ (the value for the empty assignment) will provide the minimum cost decision tree for $f$.

For any Boolean function $g$, let $|g|$ denote the size of the minimum cost decision tree consistent with $g$. For any partial assignment $\alpha$ and any variable $v$ not assigned a value in $\alpha$, let $\alpha \cup {v \gets b}$ denote the partial assignment created by assigning the value $b \in \{0,1\}$ to $v$
 to extend $\alpha$.
Let $c^j_i$ denote the cost of evaluating $x^j_i$, let $p^j_i$ denote the probability that $x^j_i = 1$, and let $q^j_i$ denote the probability that $x^j_i = 0$.

We can construct the table $P$ using dynamic programming by following these rules:

\begin{enumerate}
	\item For any complete assignment $\alpha$, the minimum size decision tree has a cost of 0, since no variables need to be evaluated to determine its value. Hence, the value $P[\alpha] = 0$.
	\item For any partial assignment $\alpha$ such that there exists a variable $v$ that has not yet been evaluated and $f_{\alpha \cup v \gets 0} = f_{\alpha \cup v \gets 1}$, then $f_\alpha$ =  $f_{\alpha \cup v \gets 0}$ and the entry $P[\alpha] =  P[\alpha \cup v \gets 0]$.
	\item For any partial assignment $\alpha$ that does not meet conditions 1 or 2, then
\[ |f_\alpha| = \min_{s^j_i \forall j} c^j_i + p^j_i \times |f_{\alpha \cup x^j_i \gets 1}| + q^j_i \times |f_{\alpha \cup x^j_i \gets 0}|. \]

Then we can fill in the entry for $P[\alpha]$ by finding the next variable $s^j_i$ that has the minimum cost testing strategy, placing it at the root of a tree and creating left and right subtrees accordingly.
\end{enumerate}

Since there are $t$ sibling classes and each can have at most $n$ variables, we can construct $P$ in time $O(n^t)$. Since $t = 2^k$, the dynamic program will run in time $O(n^{2^k})$.
\end{proof}

\begin{corollary}
\label{cor:kDNFunit}
The evaluation problem for monotone $k$-term DNF, restricted to the
uniform distribution on input $x$ and unit costs, can be solved exactly
in polynomial time for $k = O(\log \log n)$.
\end{corollary}

\begin{proof}
Under the uniformity assumption the ratios are the same for all variables. Hence, each sibling class will be evaluated as a single block and tested in an arbitrary order until either a variable evaluates to 0 or a term evaluates to 1, or until the sibling class is exhausted. Since we will evaluate each sibling class together, we can view each class as a single variable. Then we have a $k$-term DNF defined over $2^k$ variables. Let $V$ be the set of the new variables. For each $v \in V$, let $v_\ell$ denote the number of ``real'' variables in $v$.

 We can then find the optimal strategy using a dynamic programming method as before. The first two rules are as in the previous program. We will modify the third rule as follows:

For any partial assignment $\alpha$ that does not meet the first two conditions, then
\[|f_\alpha| = \min_{v \in V} 1 + \sum_{i = 1}^{v_\ell -1} \left[ \left(\frac{1}{2}\right)^i \times |f_{\alpha \cup v \gets 0}| \right] + \left(\frac{1}{2}\right)^{v_\ell+1} \times |f_{\alpha \cup v \gets 1}| + \left(\frac{1}{2}\right)^{v_\ell+1} \times |f_{\alpha \cup v \gets 0}|. \]
which follows directly from the unit costs and uniform probabilities.

The size of the table will be only $3^t$; hence we can determine the optimal testing strategy over the sibling classes in time $O(2^{2^k})$.
\end{proof}

%

\section{Expected certificate cost and optimal expected evaluation cost}
\label{sec:gapcertificate}

Some of the approximation bounds discussed in this paper are with
respect to the optimal expected cost of an evaluation
strategy, while others are in terms
of the expected certificate cost of the function, which lower
bounds the former quantity. It has been previously observed in \cite{dhkarxiv} that for arbitrary probabilities, there can be a gap of $\Omega(\log n)$ between the two measures.  In what follows, we prove that
even in the unit-cost, uniform distribution case,
the ratio between these two can be extremely large:
$\Omega(n^\epsilon)$ for any constant $\epsilon$ where $0 < \epsilon < 1$.
(Note that in the unit-cost case, both measures are at most $n$.)
We also show near-optimality of the CDNF
approximation bound of Kaplan et al. and give a gap
between two complexity measures related to decision trees
for Boolean functions.

\begin{theorem}
\label{gaptheorem}
Let $\beta$ be a constant such that $0 < \beta < 1$.
Let $f$ be a
read-once DNF formula on $n$ variables where each term
is of length $\beta \log_2 n$, and every variable appears in exactly one term.
Then $E_f[OPT]  = \Omega(n^\beta)$, and $E_f[CERT] = O(\log n)$.
\end{theorem}

\begin{proof}
Let $\log$ designate the base 2 log.
Let $k = n/(\beta \log n)$ and $m = \beta \log n$.
So $f$ is a read-once $k$-term $m$-DNF formula.

We begin by showing a lower bound on $E_f[OPT]$.
The probability that a term is equal to 1 is $1/2^m = 1/n^\beta$.
An optimal strategy for evaluating read-once formulas is known.
It works by evaluating each term in decreasing order of the term's
optimal expected evaluation cost, until either a term evaluates
to 1, or all terms have evaluated to 0~\cite{kaplanMansour-Stoc05,Greiner06}.
Since we are considering the unit-cost, uniform distribution,
case, all terms are symmetric, and terms can be evaluated in
arbitrary order.
The probability that this
strategy for evaluating $f$ evaluates at least $n^\beta$ terms,
and all evaluate to 0, is
$$R=(1-\frac{1}{n^\beta})^{(n^\beta)}$$

We first show that $R \geq e^{-2}$ for $n > 4$.
We use the standard inequality that says that for all $x$, $1 + x \leq e^x.$
From this inequality, it follows that
$$1 + \frac{2}{n^\beta} \leq e^{\frac{2}{n^\beta}} \mbox{ and hence } \frac{1}{1 + \frac{2}{n^\beta}} \geq e^{-\frac{2}{n^\beta}}.$$

It is easy to show using simple algebra that
$$1-\frac{1}{n^\beta} \geq \frac{1}{1 + \frac{2}{n^\beta}} \mbox{ for } n^\beta  \geq 2 \mbox{ and so } 1-\frac{1}{n^c} \geq e^{-\frac{2}{n^\beta}}. $$

Raising both sides to the power $n^\beta$ yields the desired result
that $R \geq e^{-2}$ for $n^\beta > 2$.

Since the probability that this optimal strategy evaluates at
least $n^\beta$ terms
is at least $e^{-2}$, and each evaluation costs at least 1, it follows
that $E_f[OPT] = \Omega(n^\beta)$.

We now upper bound $E_f[CERT]$.
By definition,
$$E_f[CERT] = \frac{n}{\beta \log n} (Prob[f=0]) + {\beta \log n}(Prob[f=1]).$$
Again using the inequality $1 + x \leq e^x$, we get that
$Prob[f=0]= (1-\frac{1}{n^\beta})^{\frac{n}{\beta \log n}} \leq e^{-\frac{n^{(1-\beta)}}{\beta \log n}}$.
Since
$\frac{n}{\beta \log n}e^{-\frac{n^{(1-\beta)}}{\beta \log n}}$
approaches 0 as $n$ approaches infinity, it is less than 1 for
large enough $n$.
It follows that
$E_f[CERT] \leq 1 + {\beta \log n} (Prob[f=1])$
for sufficiently large $n$, and since $Prob[f=1] \leq 1$,
$E_f[CERT] = O(\log n)$.
\end{proof}

For any constant $\epsilon$, where $0 < \epsilon < 1$,
there is a constant $d$ where $0 < \epsilon < d < 1$.
For large enough $n$, $n^\epsilon \log n < n^d$.
We thus have the following corollary.

\begin{corollary}
There exists a Boolean function $f$ such that
for any constant $\epsilon$, where $0 < \epsilon < 1$,\\
$E_f[OPT]/E_f[CERT] = \Omega(n^\epsilon)$.
\end{corollary}

Theorem~\ref{gaptheorem} and the above corollary
can also be interpreted as results on
average-case analogues of depth-complexity and certificate-complexity.
The depth complexity of a Boolean
function $f$ is the minimum, over all decision trees for $f$,
of the depth of that tree.  Note that the depth of the tree is the
worst-case (i.e., maximum), over all $2^n$ assignments $x$
to the variables of that function, of the number of tests (decisions)
induced by the tree on assignment $x$.
The certificate complexity of a Boolean function is
the worst-case (i.e., maximum), over all $2^n$ input assignments $x$,
of the smallest 0-certificate or 1-certificate of $f$ that
is contained in $x$.
The average depth-complexity and average certificate-complexity of a Boolean function can be defined analogously, with
worst-case replaced by average case.
Thus the average depth-complexity is equal to $E_f[OPT]$,
and the average certificate-complexity is equal to $E_f[CERT]$.

We can also use Theorem~\ref{gaptheorem} to show near-optimality of
the $O(\log kd)$ approximation bound achieved by Kaplan et al. for
monotone CDNF evaluation, with respect to $E_f[CERT]$,
the expected certificate cost
under unit costs and the uniform distribution.
The function computed by the formula in Theorem~\ref{gaptheorem} has a CNF formula
with
$(\beta \log_2 n)^{n/(\beta \log_2 n)}$ clauses.
Thus in this case $O(\log kd)$ is $O(\frac{n \log \log n}{\log n})$,
which is $O(n)$.
The strategy computed by any approximation algorithm for this problem
cannot do better than the optimal strategy,
so its expected cost must be at least $E_f[OPT]/E_f[CERT]$
times larger than $E_f[OPT]$.
It follows that the approximation $O(\log kd)$ bound of Kaplan et al.
has a matching lower bound of $\Omega((\log kd)^{\epsilon})$
(for $0 < \epsilon < 1$), with respect to the expected
certificate cost.

We do not know, however, whether it
is possible for a {\em polynomial-time} algorithm
to achieve an approximation factor much better than $O(\log kd)$
with respect to the expected cost of the optimal strategy,
$E_f[OPT]$.  We have no non-trivial lower bound for the approximation
algorithm in this case; clearly such a lower bound would have to depend
on complexity theoretic assumptions.

%

\section{ Acknowledgments}
Sarah R.~Allen was partially supported by an NSF Graduate Research Fellowship under Grant 0946825 and
by NSF grant CCF-1116594.
Lisa Hellerstein was partially supported by NSF Grants
1217968 and 0917153. Devorah Kletenik was partially
supported by NSF Grant 0917153. Tongu\c{c} \"Unl\"uyurt was partially supported by TUBITAK 2219 programme. Part of this research was performed while Tongu\c{c} \"Unl\"uyurt was visiting faculty at Polytechnic Institute of NYU and Sarah Allen was a student there.

\bibliographystyle{plain}
\bibliography{dnf}

\end{document}